\documentclass[11pt]{article}
\usepackage{caption}

\usepackage{epsfig,amsthm,amsmath,color}
\usepackage{epsfig,color}

\usepackage{fullpage}
\usepackage{framed}
\usepackage{times}

\makeatletter
\newcommand{\tinyspacing}{\let\CS=\@currsize\renewcommand
  {\baselinestretch}{.9}\tiny\CS}
\newcommand{\singlespacing}{\let\CS=\@currsize\renewcommand
  {\baselinestretch}{1}\tiny\CS}
\newcommand{\myspacing}{\let\CS=\@currsize\renewcommand
  {\baselinestretch}{1}\tiny\CS}
\makeatother

\tinyspacing

\newtheorem{theorem}{Theorem}[section]
\newtheorem{observation}{Observation}[section]

\newtheorem{lemma}[theorem]{Lemma}

\newcommand{\comment}[1]{}
\newcommand{\QED}{\mbox{}\hfill \rule{3pt}{8pt}\vspace{10pt}\par}

\comment{
\newenvironment{observation}{\mbox{}\\[-10pt]{\sc Observation.} }%
{\mbox{}\\}}

\def\Y{{\cal Y}}
\def\eps{{\epsilon}}
\def\ones{{\bf 1}}
\def\u{{\mathbf u}}
\def\v{{\mathbf v}}
\def\b{{\mathbf b}}
\def\c{{\mathbf c}}
\def\valpha{{\mathbf \alpha}}

\newcommand{\ignore}[1]{}
\newcommand{\eat}[1]{}

\addtocounter{page}{-1}
\begin{document}
\title{Understanding Fashion Cycles as a Social Choice}
\date{}
\author{
Anish Das Sarma$^1$, Sreenivas Gollapudi$^2$, Rina Panigrahy$^2$, Li Zhang$^2$\\
$^1$Yahoo! Research, $^2$Microsoft Research\\
$^1${\tt anishdas@yahoo-inc.com},  $^2${\tt\{sreenig,lzha,rina\}@microsoft.com}
}

\maketitle
\thispagestyle{empty}
\begin{abstract}

We present a formal model for studying fashion trends, in terms of
three parameters of fashionable items: (1) their innate utility; (2)
individual boredom associated with repeated usage of an item; and (3)
social influences associated with the preferences from other
people. While there are several works that emphasize the effect of
social influence in understanding fashion trends, in this paper we
show how boredom plays a strong role in both individual and social
choices. We show how boredom can be used to explain the cyclic choices
in several scenarios such as an individual who has to pick a
restaurant to visit every day, or a society that has to repeatedly
`vote' on a single fashion style from a collection.  We formally show
that a society that votes for a single fashion style can be viewed as
a single individual cycling through different choices.

In our model, the utility of an item gets discounted by the amount of boredom that has accumulated over the past; this boredom increases with every use of the item and decays exponentially when not used. We address the problem of optimally choosing items for usage, so as to maximize over-all satisfaction, i.e., composite utility, over a period of time. First we show that the simple greedy heuristic of always choosing the item with the maximum current composite utility can be arbitrarily worse than the optimal. Second, we prove that even with just a single individual, determining the optimal strategy for choosing items is NP-hard. Third, we show that  a simple modification to the greedy algorithm that simply doubles the boredom of each item is a provably close approximation to the optimal strategy. Finally, we present an experimental study over real-world data collected from query logs to compare our algorithms.
\end{abstract}
\newpage

\section{Introduction}

When an individual or a society is repeatedly presented with multiple
substitutable choices, such as different colors of cars or different
themes of musicals, we often observe a recurring shift of preferences
over time, or commonly known as {\em fashion trends}.  While some
trends are relatively easy to explain (e.g., sweater sales increasing
in the winter), some other trends may result from a variety of
factors.  In this paper, we first describe a utility model which we
think may explain such trends. Then we study the computational issues
under the model and provide simple mechanisms by which consumers may
make {\em close to optimal} decisions on which products to consume and
when, in order to maximize their {\em overall utility}.  We then
conduct experiments to show how various parameters in our model can be
estimated and to validate  our algorithm.

Understanding fashion trends are of significant academic interests as
well as commercial importance in various fields, including brand
advertising and market economics.  Therefore, there's a large body of
work in multiple disciplines -- sociology
(e.g.~\cite{herbert69,bl52}), economics (e.g.~\cite{margaret67}), and
marketing (e.g.~\cite{jmr93,william68}, on theories for evolution of
fashion.  Despite much study, there is a lack of a well accepted
theory. This is probably not surprising as what makes us like or
dislike an alternative and how that changes over time involves
economical, psychological, and social factors.  Next we describe three
such factors that influence fashion.

First, and perhaps the most basic, cause of a product becoming trendy
is its {\em utility}, intuitively capturing the value it adds to an
individual. We call this the innate utility of a product.
Second, psychologically, a person's utility of consuming a product may
be discounted by constant consumption of the same item --- as one gets
tired of existing products, he desires new and different ones.
Third, while at an individual level, we have certain inclinations
based on our tastes, these are influenced by social phenomena, such as
what we see around us, friends' and celebrities' preferences.

In this paper, we present a formal model that unifies the
aforementioned three broad categories of factors using {\em innate
utility}, {\em individual boredom}, and {\em social influence} (as
depicted in Figure~\ref{fig:fashion} and explained below).  We attempt
to construct a mathematical model for these factors and use the model
to explain the formation of fashion trends. We use the term {\em item}
to denote any product, good, concept, or object whose fashion trend we
are interested in.

\vspace{-0.2cm}
\begin{enumerate}
\setlength{\itemsep}{-0.1cm}
 \item{\bf Innate utility:} The utility of an item captures the innate
 value the item provides to an individual.  We assume it is fixed,
 independent of other influences.

 \item {\bf Individual boredom:} If we use any item for too long, we
 get bored of it, and our appreciation for it goes down.  This is
 modeled as a negative component added to the utility. This factor
 grows if one repeatedly uses the same item and fades away when one
 stops consuming the item.

 \item {\bf Social influence:} Our valuation of an item can change
 significantly by the valuation of our friends or influencing people.
 For example, when we see that many people around us like something we
 may start liking it; or we may consciously want to differ from some
 other people around us.  We model such influences as a weighted
 linear combination from other people.\footnote{Additional influence
 may come from the association of a product to things/concepts we like
 or dislike. For instance, someone may be very fond of green
 technologies or dislike things that are scary.  We may simply model
 such concepts as individuals.}
\end{enumerate}
\vspace{-0.2cm}

\begin{figure}[t]
\centering
\includegraphics[width=0.5\columnwidth]{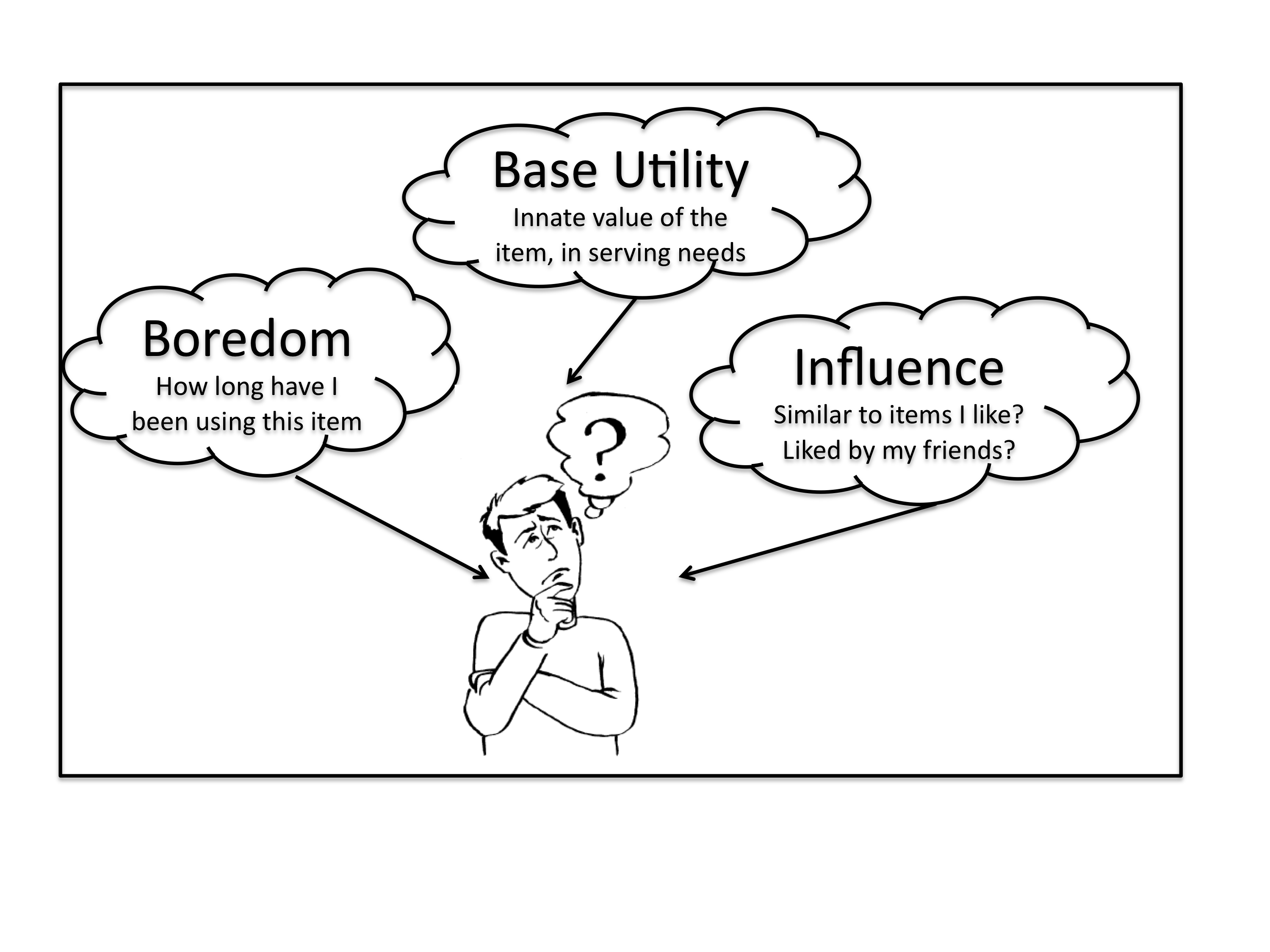}
\vspace{-1.2cm}
\caption{\label{fig:fashion}\bf Factors influencing fashion trends}
\vspace{-0.8cm}
\end{figure}

To model boredom on any item at any given time $t$, we associate with
each past usage of the item, say at time $t'$, a factor in the form of
$(1-r)^{t-t'}$ for some $r<1$.  Then the total boredom on the item
takes the sum of the this factor from all the past usage of the item.
This definition captures the intuition that the boredom grows if an
item is repeatedly used.  As we show in our experiments, such
exponential decay model matches well people's interests in songs and
movies. The utility maximization under this model, albeit NP-hard,
naturally displays cyclic patterns.  We also provide a simple strategy
to achieve near to optimal utility when the decay factor $r$ is small.

The effect of influence can be formalized using a linear model. For
example, to model social influence, consider one item and a society
consisting of $m$ people. Let $G$ denote the influence graph on these
people that is directed where each edge is labeled with a weight that
indicates the strength of this influence. A high value on an edge,
such as outgoing edges from celebrities indicates a strong outgoing
influence; on the other hand a negative value indicates a desire to
distance oneself or be different from the source node. Let $A$ denote
the corresponding influence matrix.  let $u_i(t)$ denote the utility
of the item to the $i$th person; let $\u (t)$ denote the vector of
utilities. If we assume that for each time step the influence from all
friends of a person add linearly then we may write $\u(t+1) = A
\u(t)$, which is similar to~\cite{jmr93}. Note that for stability of
this iterative powers, we should assume that its top eigenvector
has magnitude $1$..  We will show that under this influence model, we may
treat the society as an individual making choices under the effect of
boredom.

\comment{
The graph may also include nodes for other entities such as
organizations, political parties, social groups.  More generally, to
capture the effect of influence with other items/concepts the vector
$u$ may contain the utility (or liking) for all items and concepts;
correspondingly, the matrix $A$ needs to contain entries that reflect
the similarity between items/concepts.}

\medskip

\noindent\textbf{Discussion of our results.}
We argue that fashion trends can be viewed as not just the effect of
the influence of a privileged few but more as a democratic process
that churns the social boredom and channels the innate instinct for
change. Boredom is the innate psychological force that dulls the
effect of a constant stimulus over a period of time and make us look
for newer stimuli. It is well known that the mind tends to grow
oblivious to almost all types of sensations (visual, olfactory, touch,
sound) to which it is exposed for a long time. Thus the `coolness' of
a fashionable item drops over time and things that we haven't seen or
used in a long time begin to appear more `cool'.

We show how several scenarios involving individual and social choices
are essentially driven by the same underlying principles. The
individual choice may be as simple as choosing a restaurant to visit
on a particular day. Alternatively, it may be a social choice where
the market forces of a society `chooses' different fashions such as
styles for clothing, or cars. Or, a news channel is picking the front
page news article to maximize readership and has to choose from
different types of news articles, e.g., politics, natural-disaster,
celebrity gossip. Each news item may be popular or fashionable for a
period of time and then boredom sinks in and the media may switch
focus to a different event probably of an entirely different
type. Boredom is thus the single most and simplest explanation for
oscillations in individual and social choices.  This is not at all
surprising; indeed boredom is perhaps a strong influence when we make
choices such as food, clothes, fashions, governments. Social influence
no doubt plays a large part in individual choices. But when we look at
the social system as a whole the influences across individuals are
forces within the system and in the net effect it simply gives a
larger voice to the more influential individuals. We also note that influence by itself
is not sufficient to create fashion cycles. In fact, if all influences are positive then without
any boredom the system $\u(t+1) = A\u(t)$ converges to a fixed value resulting in a fixed fashion choice.

Finally, we recognize that the factors we consider are by no means comprehensive; several other `external' factors may change the values of nodes. For example a shortage of oil may increase the utility of green technologies, the strength of the edges in the graph may change, the structure of the graph may change with new node and edge formations. Our decay model for boredom and linear model for influence may be too simplistic.  Nonetheless, we believe the influence graph and boredom capture several important aspects of the underlying psychological processes that people use to value items.

\medskip

\noindent\textbf{Outline of the paper.}
All the main theoretical results achieved by this paper are presented
in Section~\ref{sec:contributions}, with proofs appearing in
Section~\ref{sec:proofs}. Section~\ref{sec:expts} presents
detailed experimental results for validating our model and algorithms;
our experiments use real-world data from Google Trends~\cite{gtrends}
on the popularity of songs and movies in the last 3 years. Related
work is presented next and we conclude in
Section~\ref{sec:conclusions}.

\eat{
The contributions of the paper are outlined as follows:
\begin{itemize}
\item Section~\ref{} formalizes our model for fashion, and studies the problem of finding an optimal strategy for consuming items. We show the non-optimality of the greedy algorithm, NP-hardness of the general version of the problem, and prove a close approximation for a slight adaptation of the greedy algorithm.

\item Section~\ref{} briefly considers switching costs between items.

\item Section~\ref{} generalizes our results for an individual to an entire society; our main result is to show that the behavior of an entire society, with social influence, can be modeled as a single individual, thus allowing us to apply techniques from the rest of the paper.

\item Section~\ref{} presents experimental results using real-world data from Google Trends~\cite{gtrends} on the popularity of songs in the last 3 years. We experimentally verify our basic model, learn its parameters, and compare multiple algorithms for choosing items. 
\end{itemize}
}

\medskip

\noindent\textbf{Related work.}
There are several theories of fashion evolution in various
communities, e.g., sociologists have modeled fashion trends as a
collection of several social forces such as differentiation,
influence, and association. While there have been several explanations
of cycles in fashion
trends~\cite{bl52,herbert69,james66,margaret67,william68}, most past
work does not offer a formal study.  We compare our work with one
notable exception~\cite{jmr93} next. The focus of our paper is on
understanding the impact of various factors---boredom, association
rules, and utility---on the fashion choices made by individuals and a
society. We explain the existence of cycles based on our formal model
of fashion, and provide algorithms for making optimal choices.

Reference~\cite{jmr93} proposed a formal model of fashion based on association rules. Intuitively, an individual's utility for an item is impacted by how similar it is to items he likes, and how dissimilar it is to items he dislikes. Further, he is influenced by the society through other individuals' preferences for various items. Consider a single item, whose consumption vector is given by $c(t)$ at time $t$. Considering the recurrence $c(t+1)  = W c(t)$, where $W$ is the weight influence matrix, \cite{jmr93} observed that if the matrix $W$ has a complex top eigenvalue (corresponding to negative influences), then the item's consumption pattern may be periodic, producing cycles in preferences. Our model of utility is similar to the consumption model in~\cite{jmr93}. However, we consider an additional parameter of boredom that is essential to explain fashion cycles in a society with non-negative influences as such a matrix $W$ always has a real top eigenvalue.

\eat{
by a person $j$ is denoted $c_j$. Let us denote by $c = (c_1,\ldots,c_n)$ the vector of consumption values for an item, which varies with time. The consumption vector at time $t$ is denoted by $c(t)$. Then, the consumption for person $1$ at time $t+1$ is given by

$c_1(t+1) = w_{1,1} c_1(t) + w_{1,2} c_2(t) +.. + w_{1,n} c_n(t)$

\noindent where $w$'s are the weighted influence due to other people. In short, we obtain the recurrence $c(t+1)  = W c(t)$, where $W$ is the weight influence matrix.
}

Some other recent work (e.g., ~\cite{ek-book}) study behavioral influences in social networks, such as in terms of information propagation. For instance, ~\cite{ikmw07} studies how two competing products spread in society,~\cite{lbk09} provides techniques for tracking and representing ``memes'', which may be used to analyze news cycles, and~\cite{lsk06} studies how recommendations propagate in a network through social influence.

The focus of our paper is on formalizing a practical theory for fashion trends with boredom, combined with utility and a simple social behavior. Therefore, for a large part of the paper we consider only a single individual and study fashion trends based on boredom, and utility. Further, in our extension to multiple individuals, we assume a linear weighting of influences from friends' preferences for particular items.

\eat{
In this paper, we add boredom as another factor that affects overall utility of using an item. Suppose a person has a choice of $n$ neighborhood for dining every night. Let $u_i$ denote the original utility obtained by eating at the $i$th restaurant. Clearly, if we assume these utilities are fixed then she would always eat the same restaurant every day. However, in reality she would get bored of eating at the same place. In this paper, we consider a model of fashion that accounts for drop in the utility of items, depending on how much recent 'memory' the person has of using a particular item. We observe that, even in the absence of negative influence, boredom explains cycles in fashion trends. We provide algorithms for determining the choice of items at any time, depending on the composite utility derived from association rules, original utility and boredom.
}

\eat{
There are several theories of how fashions evolve.
There is very little formal study of fashion cycles. [J.Market.Res 93]  proposes a formal model 
for fashion that is basically based on association. They assume that a persons utility for an item depends
on how similar it is to things he likes and dissimilar to things he dislikes. 

Our utility depends not only on our evaluation but also on the utility
of others. The valuation diffuses through the social network as
described.}

\section{Contributions of Our Study}
\label{sec:contributions}

\subsection{Modeling individual boredom}

We consider a user living in discrete time periods $0,1,\ldots$ and
consuming one item among $n$ substitutable items at each time; for
example, a person needs to decide which restaurant to go to every
night or which political party to vote for every four years.  We
assume that each item $i$ brings a base utility $v_i$ to the user.
Now if we assume that the utilities are fixed then the user would
always choose the same item with the maximum $v_i$. This would be
inconsistent with the observed common behavior of cycling among
multiple items, which we refer to as fashion cycles.  In order to
explain fashion cycles, it is necessary to model the utility
dependence of the consumptions across different time periods.

We propose a simple model in which the utility of an item at any time
$t$ is the base utility discounted by a boredom factor proportional to
the ``memory'' the person has developed by using this item in the
past.  The more the user has used the item, the more memory and
boredom is developed for the item, and consequently the less utility
the item has to the user.

We naturally assume that the memory drops geometrically over time, and
the total memory of a person is bounded.  This leads to the following
definition of memory. Let $0<r<1$ be a memory decay rate, i.e., the rate at which a
person ``forgets'' about things. Let $x_i(t)\in\{0,1\}$ indicate if
the user uses the item $i$ at time $t$. Then the memory of $i$ at time
$t$ is $M_i(t)=r\sum_{\tau=0}^{t-1} x_i(\tau) (1-r)^{t-\tau}$.  We add
the factor $r$ so that $M_i(t)\leq 1$. The boredom $b_i(t)=\alpha_i
M_i(t)$ is proportional to the memory and depends on the item.  The
utility of item $i$ is defined as $u_i(t)=v_i-b_i(t)=v_i-\alpha_i
M_i(t)$. Henceforth, we will refer to $v$ as the base utility and $\alpha$ as the boredom coefficient.

\subsection{Utility optimization with boredom}

With the above model, one natural question is to compute the choices
of the items to maximize the user's overall utility. If we allow the
user to choose at continuous time, the maximization problem becomes
relatively easy as the best way to consume an item is to do it
cyclically at regular time intervals.  However, such regular placement
may not be realizable or is hard to find. As we will show below, it is
NP-hard to compute the best consumption sequence.

We also consider the natural greedy strategy and show that the under
the greedy strategy, the utility of each item is always bounded in a
narrow band and so each item is consumed approximately cyclically.
The greedy strategy, however, may have produce a sequence giving poor
overall utility.  We provide a simple heuristics, called double-greedy
strategy, and show that it emulates the cyclic pattern of the optimal
solution on the real line and yields utility close to the optimal when
$r$ is small.

\subsubsection{Greedy algorithm}

In the greedy strategy, at each time $t$, the user consumes the item
with the maximum utility $u_i(t)$.  This strategy is intuitive and
probably consistent with how we make our daily decisions. We show that
the utility gap between any two items is small all the time.  We
provide an example to show it has poor performance in terms of utility
maximization.  Denote by $\alpha=\max_i \alpha_i$.

\begin{theorem}\label{thm:gap1}
There exists a time $T$ such that for any $t\geq T$, 
$u  \leq \max_i u_i(t)\leq u + O(r\alpha \log n)$ where $u$ 
is the unique solution to
the following system:

For all items with $f_i > 0$,  the quantity $v_i - f_i \alpha_i = \mu$;
if $f_i = 0$ then $v_i < \mu$; and
$\sum f_i =1$.
\end{theorem}

While the greedy algorithm has the nice property of keeping the
utility gap between any items small, it may produce a sequence with
poor overall utility. 

\begin{observation}
The Greedy strategy of always picking the highest utility item each day is not optimal.
\end{observation}

To see the non-optimality of greedy, simply consider two items for beverage, say ``water'' and
``soda''. Assume water has low base utility say $1$ that never changes
and zero boredom coefficient.  Soda on the other hand has high utility
say $10$ but also a high boredom coefficent say $10$. So if one drank
soda every day its utility would drop to below that of water. Observe
that the greedy strategy will choose soda till its utility drops to
that of water and then it is chosen whenever its utility rises even
slightly over $1$.  So the average utility of the greedy strategy is
close to $1$. A smarter strategy is to hold off on the soda even if it
is a better choice today so as to enjoy it even more on a later day.
Thus it is possible to derive an average utility that is much higher
than $1$. For example, we can get average utility of about $3$ by
alternating between water and soda in the above example.  Note that
the greedy algorithm produces poor performance in the above example
even for small $r$.

This naturally raises the question: what is the optimal strategy? More
importantly is there an optimal strategy that is a simple 'rule of
thumb' that is easy to remember and employ as we make the daily
choices. Unfortunately it turns out that computing the optimal
strategy is NP-hard.

\subsubsection{NP-hardness}

\begin{theorem}
Given a period $T$, target utility $U^\ast$, and $n$ items, it is NP-hard to determine
whether there exists a selection of items with period $T$ such that the total utility of the selection is at least
$U^\ast$.
\end{theorem}


\subsubsection{Double-greedy algorithm}
On the positive side we show that there is indeed a simple ``rule of
thumb'' that gives an almost optimal solution when $r$ is small. The
strategy ``double-greedy'' waits longer for items that we get bored of
too quickly.  It is a simple twist on the greedy strategy: instead of
picking the item that maximizes the utility $u_i(t) = u_i - b_i(t)$,
it picks the one which maximizes $w_i(t) = u_i - 2 b_i(t)$. Thus it
doubles the boredom of all items and then runs the greedy strategy. We show that:

\begin{theorem}\label{thm:double-greedy}
Let $\overline{U}$ denote the average utility obtained by the double
greedy algorithm and $U^\ast$ the optimal utility. Then
$\overline{U}\geq U^\ast-O(r\alpha\log n)$ where $\alpha =
\max_i\alpha_i$.
\end{theorem}

We note that when $r\to 0$, the utility produced by double greedy is close to the optimal solution.

\subsection{Fashion as a Social Choice}

A choice is a fashion, if it is the choice of a large fraction of the
society. Thus a society only supports a small number of fashions.
Industries often target one type of fashion for each market segment.
Consider a situation where the entire society consists of one fashion
market segment. We will see how in this case such a society can be compared to an
individual making choices to maximize utility under the effect of
boredom. Each individuals utiltities depend not only on his base utility and boredom but
also on the influence from other individuals.  

Consider a society of $n$ people and $m$ possible item choices. The society needs to choose one item out of these at every time step. We will study the problem of the makiing the optimal choice so as to maximize welfare. This  is applicable in the following scenarios: A business is launching the next fashion style for its market segment,  or a radio channel is broadcasting songs in a sequence to maximize the welfare to its audience. 
Let $u_{ij} (t+1)$ denote the  utility of item $i$ to  person $j$ at time $t$; let  $b_{ij}(t)$ denote the boredom value; let $\u_i(t)$ denote the vector of
utitilities to the $n$ people for item $i$, $\v_i(t)$ denote the vector of base utilities, and $\b_i(t)$ denote the vector of boredom values. In the absence
of boredom we will say  $\u_i(t+1) = A \u_i(t)$ where $A$ is the influence matrix. Accounting for boredom we will say,
$\u_i(t+1) = A\ u_i(t) - [\b_i(t+1)-\b_i(t)]$. Note that this is consistent with the case when there is only one individual
where $u_i(t+1) =  u_i(t) - [b(t+1)-b(t)]$.  
Observe that ignoring the effect of boredom we simply get the recurrence $\u_i(t+1) = A \u_i(t)$ or $\u_i(t) = A^t \v_i$. This
recurrence reflects the diffusion of influence through the social network. Note that if the largest eigenvalue of $A$ has magnitude more than $1$ then the process
will diverge and if all eigenvalues are $<1$ it will eventually converge to $0$. So we will assume the maximum eigenvalue of $A$ is has magnitude $1$.  If the gap between
the magnitude of the largest and the second largest eigenvalue is at least $\eps$ then this diffusion process converges quickly in about $\tilde \Theta(1/\eps)$ steps.
 We will focus on the case when rate of boredom $r$ is much slower than the diffusion rate (this corresponds to the case where influences spread fast and the boredom grows slowly). 
We then study the problem of making social choices of items over time so as to maximize welfare.

We will assume that $A$  is diagonlizable and has a real top eigenvalue of $1$ and all the other eigenvalues are smaller in magnitude. In that case  it is well known that for any vector $x$\, $A^t x$ converges to to a fixed point and the speed of convergence depends on the gap between the largest and second largest eigenvalue. We show that under certain conditions if $r/\epsilon$ is small then.
the choices made by the society is comparable to the choices made by an individual with appropriate base utilities and boredom coefficients. Let $W_i(t)$ denote the welfare of the society at time $t$ by choosing item $i$; then $W_i(t)/n$ 

\begin{theorem}\label{thm:socialchoice}
Consider a society with influence matrix $A$ that has largest eigenvalue $1$ and second largest eigenvalue of magnitude at most $1-\eps$.
For computing the welfare over a a sequence of social choices approximately,  such a society can be modelled as a single individual with base uitilities
$\tilde v_i$ and boredom coefficients $\tilde \alpha_i$, where $\tilde v_i = \c' \v_i $ and $\tilde \alpha_i = \c' \valpha_i$ for some vector $\c$.
Let $\tilde u_i(t)$ denote the utility of item $i$ to such an individual at time $t$.

More precisely, differences in the average utility of the society for the same sequence of choices until any time  
$|W_i(t)/n  - \tilde u_i(t)| \le \frac{r}{\eps} O(|\valpha_i|_{\infty})$ for any $t > T$  for some fixed $T$.  The $O$ notation hides factors that depends on $A$. For a real, symmetric matrix the constant is $1$
\end{theorem}

\section{Technical details}
\label{sec:proofs}

\subsection{Individual choice}

The following Lemma is used in the proof of Theorem~\ref{thm:gap1}.

\begin{lemma}\label{lem:mem}
$\sum_i M_i(t)\leq 1$, and $\sum_i M_i(t) \to 1$ for large $t$. When $t=\Omega(1/r)$, $\sum_i M_i(t)=1-O(\exp(-tr))$.
\end{lemma}
\begin{proof}
Observe that the memory scales down by a factor of $1-r$ each time step; exactly one item is picked and $r$ is added to its memory.
So $\sum_i M_i(t+1) = (1-r)\sum_i M_i(t)  + r$. This recurrence gives,  
$\sum_i M_i(t) = (1-r)^{t} \sum_i M_i(0) + r \sum_{j=0}^t (1-r)^j +(1-r)^{t}( \sum_i M_i(0)  -1) $. Since $M_i(0)=0$, $\sum_i M_i(t)\leq 1$. 
Observe also that after t = $\Omega(1/r)$ steps this becomes $1+ O(exp(-tr))$
\end{proof}

We are now ready to prove Theorem~\ref{thm:gap1}.
\begin{proof}(\textbf{Theorem~\ref{thm:gap1}})
To see that the solution to the given system is unique, note that $f_i = (\frac{v_i -u}{\alpha})^{+}$ (where $x^{+}$ denotes $max(x,0)$, and so
$\sum_i  (\frac{v_i -u}{\alpha})^{+} = 1$. This must have a unique solution as $\sum_i (\frac{v_i -u}{\alpha})^{+}$ is decreasing function of $u$ and strictly decreasing as long as the sum is positive.
Let $u$ denote the solution to the above system. 

We now show $\max_i u_i(t)\geq u$ for any $t$.  This is done by
contradiction. Suppose that for all $i$ $u_i(t)<u$.  We have that
$\sum_i(\frac{v_i-u}{\alpha_i})^{+}<
\sum_i \frac{v_i-u_i(t)}{\alpha_i}$. But $\sum_i \frac{v_i-u_i(t)}{\alpha_i} = \sum_i M_i(t)\leq 1$. We have that $\sum_i(\frac{v_i-u}{\alpha_i})^{+}<1$, a
contradiction.

Let $S_g$ denote the set of all the items ever picked by the greedy
algorithm. Let $T$ be the time by which each item in $S_g$ has been
used at least once. By Lemma~\ref{lem:mem}, after some steps $\sum_i
M_i(t)$ converges to arbitrarily close to $1$. Lets assume for
simplicity of argument that it is exactly $1$ with sufficiently large
$T$. To show the upper-bound on $\max_i u_i(t)$, we show that for
$t\geq T$ and any $i\in S_g$, $\max_j u_j(t)-u_i(t)=O(\alpha r\log
n)$.

Denote by $x(t)$ the item that has the maximum utility at time $t$. 
It suffices to show that $u_{x(t)}(t)\leq u_i(t)+O(\alpha r \log n)$.
We recursively compute a decreasing sequence of $t_j$ as follows. Let
$t_1 = t$. For $j>1$, suppose we have computed $t_{j-1}$. Let $S_{j-1}
=
\{x(t_1),x(t_2),\cdots,x(t_{j-1})\}$. Now let
$t_j=\max_{t'<t_{j-1}, x(t')\notin S_{j-1}} t'$. We stop when there is
$k$ such that $x(t_k)=i$.  Since $t>T$, the process is guaranteed
to stop.  By the above construction, we know only items in
$S_{j-1}$ are picked by the greedy algorithm in the interval
$[t_{j}+1, t_1]$.  For any $S\subseteq \{1,\cdots, n\}$, let
$A(S,t)=\sum_{\ell\in S}\frac{u_\ell(t)}{\alpha_\ell}$, and $B(S)=
\sum_{\ell\in S}\frac{1}{\alpha_\ell}$.  We will show that for $1<j\leq k$.
\begin{equation}\label{eq:1}
A(S_{j-1},t_{j-1})\leq \frac{B(S_{j-1})}{B(S_j)}A(S_j, t_j)+r\,.
\end{equation}

First observe that
\begin{equation}\label{eq:2}
A(S,t)= \sum_{\ell\in S}\frac{u_\ell(t)}{\alpha_\ell}
=\sum_{\ell\in S}\frac{V_\ell -\alpha_\ell M_\ell(t)}{\alpha_\ell}
=\sum_{\ell\in S}\frac{V_\ell}{\alpha_\ell} -\sum_{\ell\in S} M_\ell(t)\,.
\end{equation}

(\ref{eq:1}) follows from the following claims.

\paragraph{Claim 1.} $A(S_{j-1},t_{j-1})\leq A(S_{j-1},t_j+1)$. 
\begin{proof}
Since any item picked by the greedy algorithm in $[t_j+1,t_{j-1}]$ is
in $S_{j-1}$, we have that for $t'\in [t_j+1,t_{j-1})$, $\sum_{\ell\in S_{j-1}}M_{\ell}(t'+1) = (1-r)\sum_{\ell\in S_{j-1}}M_{\ell}(t')+r
\geq \sum_{\ell\in S_{j-1}}M_{\ell}(t')$. The last inequality is by $\sum_{\ell} M_{\ell}(t')\leq 1$.  Therefore 
$\sum_{\ell\in S_{j-1}}M_{\ell}(t_{j-1})\geq
\sum_{\ell\in S_{j-1}}M_{\ell}(t_{j}+1)$. 
By (\ref{eq:2}), we have $A(S_{j-1},t_{j-1})\leq A(S_{j-1},t_j+1)$. 
\end{proof}

\paragraph{Claim 2.} $A(S_{j-1},t_j+1)\leq A(S_j, t_j)+r$. 
\begin{proof}
Since $t_j\notin S_{j-1}$ is the item picked by the greedy algorithm
at $t_j$, $\sum_{\ell\in S_{j-1}}M_{\ell}(t_j+1) = (1-r)\sum_{\ell\in S_{j-1}}M_{\ell}(t_j)$. Thus
$\sum_{\ell\in S_{j-1}}M_{\ell}(t_j+1)-\sum_{\ell\in S_{j-1}}M_{\ell}(t_j)
= r \sum_{\ell\in S_{j-1}}M_{\ell}(t_j) \leq r$. Again by (\ref{eq:2}), we have
 $A(S_{j-1},t_j+1)\leq A(S_j, t_j)+r$. 
\end{proof}

\paragraph{Claim 3.} $A(S_{j-1},t_j)\leq \frac{B(S_{j-1})}{B(S_j)}A(S_j, t_j)$.\begin{proof}
Immediately follows from $u_{t_j}(t_j)\geq u_{\ell}(t_j)$ for $\ell\in S_{j-1}$. 
\end{proof}
\medskip

Repeating (\ref{eq:1}), we have that
\begin{eqnarray*}
&&A(S_{j-1},t_{j-1})\\
&\leq& \frac{B(S_{j-1})}{B(S_j)}A(S_j, t_j)+r\leq \frac{B(S_{j-1})}{B(S_j)}\left(\frac{B(S_{j})}{B(S_{j+1})}A(S_{j+1}, t_{j+1})+r\right)+r\\
&=& \frac{B(S_{j-1})}{B(S_{j+1})}A(S_{j+1}, t_{j+1})+r\cdot\frac{B(S_{j-1})}{B(S_j)}+r \cdots \leq \frac{B(S_{j-1})}{B(S_k)}A(S_k, t_k)+r\cdot \sum_{\ell={j-1}}^{k-1} \frac{B(S_{j-1})}{B(S_\ell)}\,.
\end{eqnarray*}

Hence, we have that 
\[u_{x(t_1)}(t_1) = \alpha_1 A(S_1,t_1)\leq\alpha_1\left(\frac{B(S_1)}{B(S_k)}A(S_k, t_k)+r\cdot\sum_{\ell=1}^{k-1} \frac{B(S_1)}{B(S_\ell)}\right)=\frac{1}{B(S_k)}A(S_k, t_k)+r\cdot\sum_{\ell=1}^{k-1} \frac{1}{B(S_\ell)}\,.
\]
Since $i = x(t_k)$, for any $i'$, $u_{i'}(t_k)\leq u_i(t_k)$. 
Therefore  $A(S_k,t_k)=\sum_{j\in S_k}\frac{u_j(t_k)}{\alpha_j}
\leq u_i(t_k)\sum_{j\in S_k}\frac{1}{\alpha_j}=u_i(t_k)B(S_k)$.
By that $\alpha = \max_i \alpha_i$, we have 
$B(S_\ell) \geq \ell/\alpha$.  Hence
\[
u_{x(t_1)}(t_1) \leq \frac{1}{B(S_k)}A(S_k, t_k)+r\cdot\sum_{\ell=1}^{k-1} \frac{1}{B(S_\ell)} \leq u_i(t_k)+\alpha r \sum_{\ell=1}^{k-1} 1/\ell = u_i(t_k)+ O(\alpha r \log n)\,.
\]
Since item $i$ is not used during the interval of $[t_k+1,t_1]$, we have 
$u_i(t_k+1)\leq u_i(t_1)$, and hence $u_i(t_k)\leq u_i(t_k+1)+\alpha_i r \leq u_i(t_1)+\alpha_i r$. Therefore we have that,
$\max_j u_j(t_1) = u_{x(1)}(t_1)\leq u_i(t_1)+ O(\alpha r \log n)$.

On the other hand, we know that there exists $i\in S_g$ such that $u_i(t)\leq
u$ because otherwise it would be the case that
$\sum_i(\frac{v_i-u}{\alpha_i})^+>\sum_i\frac{v_i-u_i(t)}{\alpha_i}=\sum_i
M_i(t)\approx 1$, a contradiction.  Hence $\max_j u_j(t)=u+O(\alpha r
\log n)$.
\end{proof}

\subsection{NP-hardness of item selection}

A selection $Y$ is periodic with period $T$, if for any $t$,
$y(t+T)=y(t)$, where $y(t)$ is the item chosen at time $t$.  Clearly,
in a periodic selection, the utility of the item chosen at time $t$ is
the same as the one chosen at time $t+T$.  For utility maximization,
it suffices to consider those items chosen in $[0,T)$.  Let
$U(Y)=\sum_{t=0}^{T-1} u_{y(t)}(t)$ denote the total utility of $Y$ in
$[0,T)$.

\begin{theorem}\label{thm:nphard}
It is NP-hard to decide, given $T$, and $U^\ast$, and $n$ items,
whether there exist a assignment $Y$ with period $T$ such that
$U(Y)\geq U^\ast$.
\end{theorem}

\begin{proof}
The reduction is from the Regular Assignment Problem and is detailed in the appendix.
\end{proof}

\subsection{Optimality of double-greedy algorithm}

Using the exactly same argument in the proofs of
Theorem~\ref{thm:gap1}, we have that

\begin{lemma}\label{lem:gap2}
There exists a time $T$ such that for any $t\geq T$, 
$\mu  \leq \max_i w_i(t)\leq \mu + O(\alpha \log n r)$ where $\mu$ 
is the unique solution to
the following system:

For all items with $f_i > 0$,  the quantity $v_i - 2 f_i \alpha_i = \mu$;
if $f_i = 0$ then $v_i < \mu$; and
$\sum f_i =1$.
\end{lemma}

By using the above theorem, we can prove Theorem~\ref{thm:double-greedy} as follows.
\begin{proof}(\textbf{Theorem~\ref{thm:double-greedy}})
For $0< f\leq 1$, write $\Delta(f)
=r\cdot\frac{(1-r)^{1/f}}{1-(1-r)^{1/f}}$.  

Let $U^\ast$ be the optimal value of the following program.
\begin{equation}\label{eq:0}
\max U = \sum_i f_i(v_i - \alpha_i f_i) \quad \mbox{s.t.}\quad \sum_i f_i\leq 1 \quad \mbox{and $f_i\geq 0$.}
\end{equation}

Let $OPT$ denote the optimal average utility. We have that $OPT\leq
U^\ast+\alpha r$.  This is by observing that for any $0<f<1$, placing
an item $1/f$ apart gives an upper bound on the utility of consuming
the item with frequence $f$.  The bound is $v-\alpha
\Delta(f) \leq v-\alpha(r-f)$ by observing that
 $\Delta(f)>r-f$.

The objective of (\ref{eq:0}) is maximized when there exists $\lambda$
such that $\frac{\partial U}{\partial f_i}=\lambda$ for $f_i>0$ and
$\frac{\partial U}{\partial f_i}<\lambda$ for $f_i=0$, and $\sum_i f_i
= 1$.  Since $\frac{\partial U}{\partial f_i} = v_i - 2\alpha_i f_i$,
$\lambda$ is exactly the same as $\mu$ in the statement of
Lemma~\ref{lem:gap2}. This explains the intuition of the double greedy
heuristics --- it tries to equalize the marginal utility gain of each
item. Denote the optimal solution by
$f_i^\ast$. Then for $f_i^\ast>0$, $v_i - 2\alpha_i f_i^\ast = \mu$. Hence,
\[U^\ast = \sum_i
f_i^\ast(v_i-\alpha_i f_i^\ast) = \sum_i f_i^\ast
(\mu+\alpha_i f_i^\ast ) = \mu+\sum_i \alpha_i {f_i^\ast}^2\,.\]

Let $k_i$ denote the number of times item $i$ is used in $[0,T]$ by
the double-greedy algorithm, and $f_i = k_i/T$.  Let $\overline{M}_i$
denote the average memory on $i$ at the times when $i$ is picked.
Then we have that
\begin{eqnarray}
\overline{U}&= &\sum_{t=0}^T u_{x(t)}(t)/T=\sum_i \sum_{x(t)=i} u_i(t)/T=\sum_i \sum_{x(t)=i} (w_i(t)+\alpha_i M_i(t))/T\nonumber\\
&\geq&\sum_i \sum_{x(t)=i} (\mu+\alpha_i M_i(t))/T \quad\mbox{(by Lemma~\ref{lem:gap2}, $w_i(t)\geq \mu$)}\nonumber\\
&\geq&\mu + \sum_i \alpha_i f_i \overline{M}_i\,.\label{eq:3}
\end{eqnarray}
Write $\delta = \alpha r \log n$.  By Lemma~\ref{lem:gap2},
$w_i(t)=\mu + O(\delta)$ for each $i,t$.  We will show that
\vspace{-0.2cm}
\paragraph{Claim 1. $\alpha_i f_i = \alpha_i f_i^\ast - O(\delta)$.}
Observe that for any item $i$ which is picked $k_i$ times in $[0,T]$,
$\min_{0\leq t\leq T} M_i(t)\leq \Delta(k_i/T)\leq k_i/T=f_i$.  Hence, $\max_t w_i(t)\geq v_i -
2\alpha_iM_i(t) \geq v_i - 2\alpha_i f_i $. On the other hand, $w_i(t)
= \mu + O(\delta)$.  
We have $v_i - 2\alpha_i f_i = \mu+O(\delta)$. But $\mu = v_i -
2\alpha_i f_i^\ast$. Therefore $\alpha_i f_i \geq \alpha_i f_i^\ast - O(\delta)$. 
\vspace{-0.2cm}
\paragraph{Claim 2. $\alpha_i \overline{M}_i = \alpha_i f_i^\ast -O(\delta)$.}
Since $v_i - 2\overline{M}_i \leq \mu + \delta$, we obtain the bound by following the same argument as in the proof of Claim 1.  Now, plugging both claims into (\ref{eq:3}), we have that
\[\overline{U}\geq \mu + \sum_i \alpha_i (f_i^\ast -\frac{O(\delta)}{\alpha_i})^2 \geq \mu + \sum_i \alpha_i {f_i^\ast}^2 - O(\delta)
= U^\ast - O(\delta)
 = OPT-O(\delta)\,.\]
This last equality follows from $U^\ast = OPT - \alpha r$. 
This completes the proof. 
\end{proof}

\subsection{Social Choice is equivalent to individual choice}

Let $\ones$ denote the vector with all coordinates set to $1$ and $\valpha_i$ denote the vector of boredom coefficients $\alpha_{ij}$.
\vspace{-0.3cm}
\paragraph{Observation.}
For any diagonolizable matrix $A$ with largest eigenvalue $1$ and the second largest eigenvalue is at most $1-\eps$, there is a vector $c$ so that.
$\ones' A^t x - c' x \le  (1-\eps)^t \sqrt n O(|x|_2)$. The $O$ notation hides factors that depends on $A$. For a real, symmetric matrix the constant is $1$.
\begin{proof}
We will sketch the proof for real symmetric matrices. The same idea holds for non-symmetric matrices.
If $p_1, \ldots, p_n$ denote the eigenvectors of $A$ and $1=\lambda_1, \ldots, \lambda_n$ denote the eigenvalues then
$A^t x = \sum_j \lambda^t_j p_i v'_i x = p_1 v'_1 x +  \sum_{j>1} {\lambda^t}_j p'_i p_i x$.
Now,|$\sum_{j>1} \lambda^t_j p'_i p_i x|_2 \le (1-\eps)^t |x|_2$.
So $|\ones'(A^t x - p_1 p'_1 x)| \le  |\ones|_2 (1-\eps)^t |x|_2 = \sqrt n (1-\eps)^t |x|_2$
Setting $c = \ones' v_1 v'_1$ completes the proof.
\end{proof}

We will now prove theorem~\ref{thm:socialchoice}
\begin{proof}(\textbf{Theorem~\ref{thm:socialchoice}})
Let $\Delta \b(t)$ denote $\b(t) -\b(t-1)$.
Now $\u_i (t) =  A \u_i(t-1) + \Delta \b(t)$. This gives, $\u_i(t) =  A^t \v_i  + \sum_{j=0}^{t-1} A^j \Delta \b (t) $.
Note that $W_i(t)$ = $\ones' \u_i(t) =\ones'  A^t \v  + \sum_{j=0}^{t-1} \ones' A^j \Delta\ b (t-j) $

Note $\Delta b_{ij} (t) = \alpha_{ij} ((1-r) M_{ij}(t)  + r I_i (t) - M_{ij} (t)) =  \alpha_{ij} r( x_i (t) - M_{ij} (t))$.
So $|\Delta b_{ij} (t)|_2 \le r |\valpha_{i}|_2$.


Now  $|\ones' A^t \v_i - c' \v_i| \le (1-\eps)^t \sqrt n O(|\v_i|_2)$. For $t> (1/\eps)\Omega( \log( nr|\v_i|_2) $, this is at most $r$.
Also  $|\ones A^j \Delta \b_i (t-j) - c' \Delta \b_i (t-j)| \le  (1-\eps)^j  \sqrt n O( |\Delta  \b_i (t-j)|_2) \le (1-\eps)^j  \sqrt n rO( | \valpha_i|_2) $.
So,  $|\sum_j \ones A^j \Delta \b_i (t-j) - \sum_j c' \Delta \b_i (t-j)| \le (r/\eps) \sqrt n O(|\valpha_i|_2)$.
So,  $|\sum_j \ones A^j \Delta \b_i(t-j) - c'\ b_i (t)| \le (r/\eps)  \sqrt n O(|\valpha_i|_2)$.
Therefore $|W_i (t) - (c'\v_i - c' \b_i(t))| \le (r/\eps) O(\sqrt n  |\valpha_i|_2)$. Dividing by $n$ completes the proof.
\end{proof}
\vspace{-0.3cm}

\section{Experiments}
\label{sec:expts}
In this section, we provide experimental results to study the techniques presented in the paper. Our primary objectives is to  evaluate the quality of greedy and double-greedy algorithms for choosing items based on utility and boredom parameters estimated from the real data.
\eat{
\begin{itemize}
\setlength{\itemsep}{-0.1cm}
\item Verify whether boredom indeed plays a role in reducing the popularity (overall utility) of items. 
\item Compare the greedy and double-greedy algorithms for choosing items based on utility and boredom parameters estimated from the real data
\end{itemize}
}

\subsection{Setup}

We obtain data on the popularity of songs and movies from Google Trends~\cite{gtrends}. We collected weekly aggregate counts from query logs for popular songs from the last 3 years. Similar data was collected for popular movies. While the popularity of songs and movies depends on additional factors such as awards won by an album or a movie, our goal was to perform a controlled experiment only based on overall utility and boredom. Therefore, for each item we collected weekly aggregate counts starting from the highest peak in logs till there was an ``artificial peak'' due to an external event such as an award.  Further, we compare the utility obtained by our model with a baseline in which the user selects an item simply based on its utility without any discounting from boredom.  We describe how we compute the values of $\alpha$, $v$, and $r$ in the appendix.

\begin{table*}[t]
\footnotesize
\begin{tabular}{cc}
\begin{minipage}{3.5in}
\begin{tabular}{|c|c|c|c|}
\hline
{\bf Song} & $v$ & $\alpha$ & $r$ \\ \hline
The Climb & 12.3 &   9.9  &   0.097 \\
Lucky    &    2.6     &  1.58   &   0.114 \\ 
Snow Patrol - Chasing Cars  &     10.7   &   6.8   &    0.127 \\
I know you want me      &  7.95  &    6.5     &  0.077 \\
Viva la vida   &   12.4    &  9.1     &  0.16 \\
Stop and stare  &  10.5   &   9.4    &   0.092 \\
Disturbia      &   8      &   7.2     &  0.092 \\
Pocket full of sunshine &  7.6     &  6.3  &     0.14 \\
Supernatural superserious   &      24.2 &     22   &     0.15 \\
One step at a time      &  9.35    &  8.5     &  0.075 \\
\hline
\end{tabular}
\caption{\label{tab:ual1}$v, \alpha, r$ for the set of songs.}
\end{minipage}
&
\begin{minipage}{3.5in}
\begin{tabular}{|c|c|c|c|}
\hline
{\bf Movie} & $v$ & $\alpha$ & $r$ \\ \hline
Godfather  &     6.15 &  5.15  &  0.123 \\
Hancock & 9.6  &   8.8  &   0.128 \\
The Bucket List & 13.1 &   11.8  &  0.102 \\
Quantum of Solace   &    29.8  &  29  &    0.111 \\
Tropic Thunder &  25.6  &  24.8 &  0.082 \\
\hline
\end{tabular}
\caption{\label{tab:ual2} $v, \alpha, r$ for the set of movies.}
\end{minipage}
\end{tabular}
\vspace{-0.5cm}
\end{table*}

Table~\ref{tab:ual1} shows the $v$, $\alpha$, and $r$ values for a set of 10 songs used in our experiments while the corresponding data for the movie data set is shown in Table~\ref{tab:ual2}; here we allow different values of $r$, but we notice that all $r$-values within the domain of songs and movies are similar.

\subsection{Results}

We ran a set of experiments to verify the effectiveness of the greedy and double-greedy heuristics. We ran the experiments over $100000$ steps for both the data sets.  The average utility obtained by the user for both the data sets was computed and is shown in Table~\ref{tab:au}. We also show results for the baseline approach that always picks the same item with the highest base utility. Tables~\ref{tab:aus} and ~\ref{tab:aum} illustrate the average utility obtained by the user over the selected songs and movies respectively. The corresponding normalized frequencies are shown in parenthesis. As expected, in the baseline case where the user selects an item according to its base utility, the movie {\tt Quantum of Solace} (with a base utility of $29.8$) is always selected while in the case of songs, the song {\tt supernatual superserious} (with a utility of $24.2$) is selected.  Unsurprisingly, the average utility discounting boredom for this case is very low (see Table~\ref{tab:au}).

{\centering
\footnotesize
\begin{minipage}{0.45\textwidth}
\centering
  \begin{tabular}{|c|c|c|c|}
  \hline
  {\bf Dataset} & Greedy & Double-Greedy & Baseline \\ \hline
  Songs & $11.94$ & $13.53$ & $5.62$ \\
  Movies & $16.12$ & $17.30$ & $4.25$ \\
  \hline
  \end{tabular}
  \captionof{table}{Average utility over $100000$ time steps.}
  \label{tab:au}
\end{minipage}
\begin{minipage}{0.45\textwidth}
  \centering
  \includegraphics[width=2.5in]{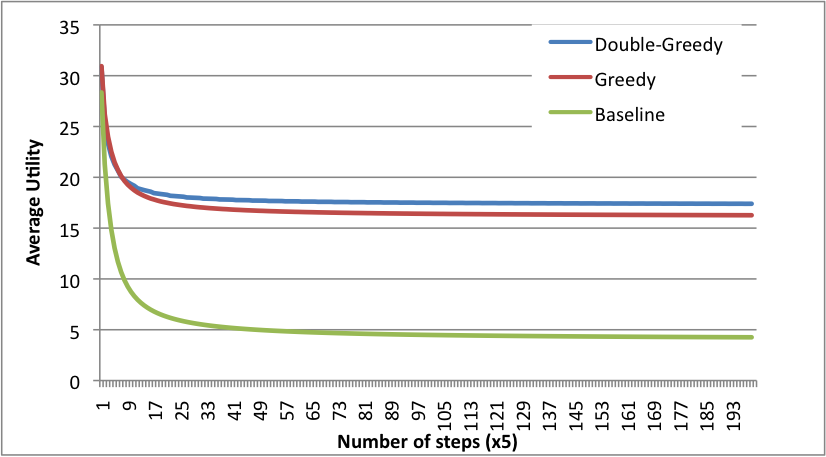}
  \captionof{figure}{Change in average utility over time.}
  \label{fig:aut}
\end{minipage}
}

\begin{table*}
\footnotesize
\begin{tabular}{cc}
\begin{minipage}{3.5in}
\begin{tabular}{|c|c|c|}
\hline
{\bf Song} & Greedy & Double-Greedy \\ \hline
The Climb & $11.17 (0.17)$ & $11.11 (0.17)$ \\
Snow Patrol - Chasing Cars  & - (0) & $10.24 (0.12)$ \\
Viva la vida & $11.61 (0.17)$ & $11.24 (0.21)$ \\
Stop and stare  & - (0) & $10.08 (0.08)$ \\
Supernatural superserious & $12.22 (0.67)$ & $17.52 (0.40)$ \\
\hline
\end{tabular}
\caption{\small Avg. utilities (frequencies) for selected songs.}
\label{tab:aus}
\end{minipage}
&
\begin{minipage}{3.5in}
\begin{tabular}{|c|c|c|}
\hline
{\bf Movie} & Greedy & Double-Greedy \\ \hline
Hancock & - (0) & $9.56 (0.04)$ \\
The Bucket List & - (0) & $11.40 (0.20)$ \\
Quantum of Solace & $16.41 (0.55)$ & $20.37 (0.40)$ \\
Tropic Thunder & $15.77 (0.45)$ & $18.01 (0.36)$ \\
\hline
\end{tabular}
\caption{\small Avg. utilities (frequencies) for selected movies.}
\label{tab:aum}
\end{minipage}
\end{tabular}
\vspace{-0.5cm}
\end{table*}

\medskip

In another experiment, we measured the change in the average utility with time. Figure~\ref{fig:aut} illustrates the change in average utility as the user selects different items at each time step for movies.  Naturally, the utility is highest at the very beginning as the user picks an item with the highest base utility and decreases subsequently as she picks items with highest discounted utility at each time step.

\section{Future Work}
\label{sec:conclusions}

As we mentioned, our model is by no means comprehensive. For example,
boredom may come from consuming similar items, or there may be a cost
when switching from item to item.  Taking into account these factors
raises some interesting algorithmic issues.  Fully incorporating these
extensions is left as future work.

\section{Acknowledgements}
We thank Atish Das Sarma for useful discussions.

\newpage

\bibliographystyle{abbrv}
\bibliography{fc}

\appendix
\section{Computing the model parameters}

Figure~\ref{fig:trend} shows the trend observed for a specific song from our dataset, {\tt I Know You Want Me}, over a 45-week period starting August 2, 2009. The first natural observation we make is that the total number of queries do indeed display a steady decline, which we attribute to boredom. From the data, we use the maximum count as the peak utility, $v_{peak}$, and let the final count be denoted $v_{final}$. We set $\alpha=v_{peak}-v_{final}$. Let $X(t)$ denote the aggregate count for the week $t$, we obtain the boredom parameter $r$ using the following equation:

$e^{-r t} = 1 - \frac{v_{peak}-X(t)}{v_{peak}-v_{min}}$

\noindent We plot $r t = - \ln (1  - \frac{v_{peak}-X(t)}{v_{peak}-v_{min}})$, and fit a linear line on the resulting curve and obtain $r$ from the slope.  Figure~\ref{fig:linear} shows the curve for {\tt I Know You Want Me}, from which we obtain the $r$ value.

\begin{figure*}[t]
\centering
\begin{tabular}{cc}
\begin{minipage}{3.0in}
\includegraphics[width=3.0in]{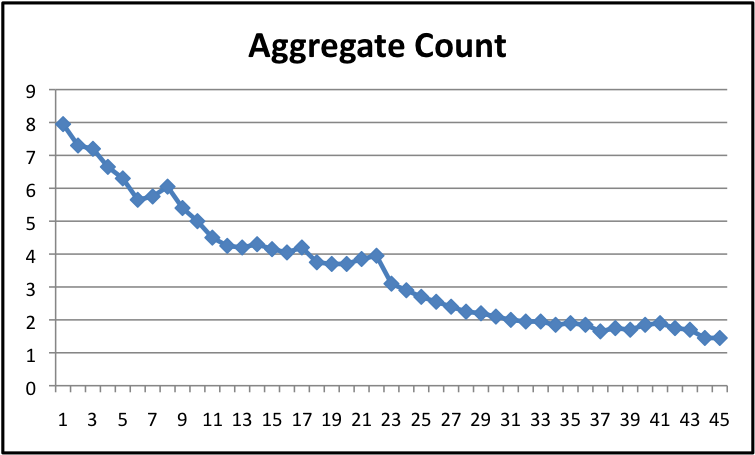}
\caption{Weekly aggregate query counts for {\tt I Know You Want Me} for a 45-week period from Google trends.}
\label{fig:trend}
\end{minipage}
&
\begin{minipage}{3.0in}
\centering
 \includegraphics[width=3.0in]{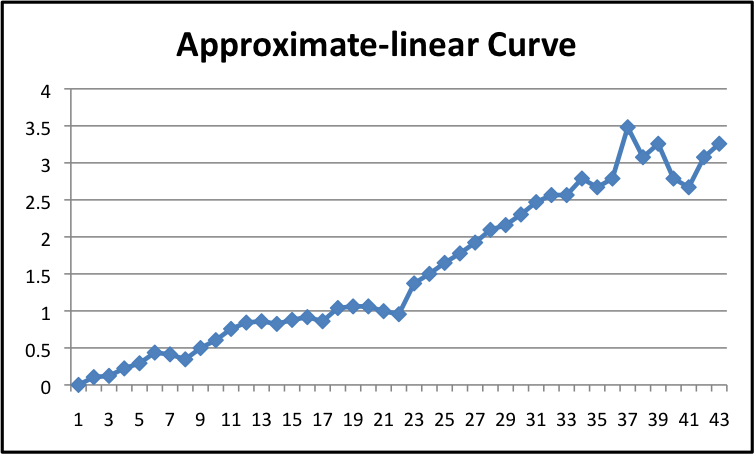}
\caption{Approximate linear trend for {\tt I Know You Want Me}, with slope giving $r$.}
\label{fig:linear}
\end{minipage}
\end{tabular}
\end{figure*}

\section{NP-hardness of item selection}

{\bf Restatement of Theorem~\ref{thm:nphard}:}
It is NP-hard to decide, given $T$, and $U^\ast$, and $n$ items,
whether there exist a assignment $Y$ with period $T$ such that
$U(Y)\geq U^\ast$.

\begin{proof}
The reduction is from the following problem.

\paragraph{Regular assignment problem (RAP).}
Given positive integers $p_1, p_2, \cdots, p_n$, determine if there
exists a sequence $y_0,y_1,\cdots$ where $y_t\in\{0,1,\cdots,n\}$ such
that for any $i\neq 0$, two consecutive appearances of $i$ in the
sequence are exactly $p_i$ apart.

It is shown in~\cite{barnoy-02} that the regular assignment problem is
NP-complete.  Note that for RAP, a regular assignment exists if and
only if it does so on a cycle with length $T=\prod_{i} p_i$.  We will
now reduce it to the optimal fashion selection problem.

Given $p_1,\cdots, p_n$, we create $n+1$ items such that a regular
assignment, if exists, maximizes the utility of any periodic selection
with period $T$.  Hence we can reduce RAP to the optimal selection
problem.  Item $0$ is a special item with $v_0 = 1$ and $\alpha_i=0$.
For $1\leq i\leq n$, we assign $v_i=\frac{2T}{p_i}$ and $\alpha_i =
1$.  Further let $r_i=1/T$ for $1\leq i\leq n$.  We claim that there
exists $U^\ast$ and $\epsilon\geq 1/T^2$ such that for a regular
assignment $Y$, $U(Y)\geq U^\ast$, and $U(Y)<U^\ast-\epsilon$
otherwise.

Consider the case when there is only item and when the selections are
made on the real line. Given $T$ and an item with parameters
$v,\alpha,r$, let $\Y_k(v,\alpha,r)$ be the set of all the selections
which have period $T$ and choose the item exactly $k$ times on the
real interval $[0,T)$.  Denote by
$U_k(v,\alpha,r)=\max_{Y\in\Y_k(v,\alpha,r)} U(Y)$ and $\delta
U_k(v,\alpha,r)= U_{k}(v,\alpha,r)-U_{k-1}(v,\alpha,r)$. The correctness
of the reduction follows from the following claims.

\paragraph{Claim 1.}  $U_k(v,\alpha,r)=kv-k\alpha\frac{(1-r)^{T/k}}{1-(1-r)^{T/k}}$, and the maximum is achieved with the regular assignment.

\paragraph{Claim 2.} For $1\leq v\leq n$, $U_k(v_i,1,1/T) = kv_i - (k^2-\frac{1}{2}k+\frac{1}{12}+o(1/k))$, and $\delta U_k(v_i,1,1/T) = v_i - (2k-\frac{3}{2}+o(1/k))$. 

\paragraph{Claim 3.}  For any non-regular \textbf{integral} selection $Y\in\Y_k(v_i,1,1/T)$,  $U(Y)<U_k(v_i,1,1/T)-1/T^2$. 

\medskip

Claim 1 holds because the total memory is minimized when the $k$
assignments are regularly spaced.  Claim 2 is a direct consequence of
Claim 1 by Taylor expansion on those particular parameters. Claim 3
follows by comparing the memory caused by adjacent items between
regular and non-regular assignments.

From Claim 2, we can see that $\delta U_k(v_i,1,1/T) \geq 3/2$ for
$k\leq T/p_i$ and $<0$ for $k\geq T/p_i+1$ for $1\leq i\leq n$, and
$\delta U_k(v_0,0,1/T)=1$.  Combining it with Claim 3, we have that
the utility gap between a regular and non-regular assignment is at
least $1/T^2$.  Therefore the reduction is correct and can be done in
polynomial time. 

\end{proof}

\end{document}